\documentclass[final,5p,times,twocolumn]{elsarticle}

\usepackage{amssymb}   
\usepackage{amsthm}    
\usepackage{amsmath}
\usepackage{enumerate}
\usepackage{graphicx}
\usepackage{xcolor}

\def\K{\mathcal{K}}
\def\Kinf{\K_{\infty}}
\def\KL{\mathcal{KL}}
\def\R{\mathbb{R}}
\def\N{\mathbb{N}}
\def\X{\mathcal{X}}
\def\E{\mathcal{E}}
\def\U{\mathcal{U}}
\def\D{\mathcal{D}}
\def\S{\mathcal{S}}

\def\dfn{:=}
\def\ms{\medspace}
\def\comp{{\scriptstyle\,\circ}\,}
\def\Ki{\K_\infty}
\def\T{\mathcal{T}}
\def\mer{\hfill $\circ$}

\def\qed{$\hfill\blacksquare$}

\newtheorem{theorem}{Theorem}[section]
\newtheorem{definition}[theorem]{Definition}
\newtheorem{remark}[theorem]{Remark}
\newtheorem{lemma}[theorem]{Lemma}
\newtheorem{aplemma}{Lemma}

\theoremstyle{remark}
\newtheorem{claim}{Claim}

\journal{System \& Control Letters}

\begin{document}

\begin{frontmatter}



\title{Characterization of semiglobal stability properties for discrete-time models \\ of non-uniformly sampled nonlinear systems}

\author{Alexis J. Vallarella}
\author{Hernan Haimovich}

\address{Centro Internacional Franco-Argentino de Ciencias de la Informaci\'on y de Sistemas (CIFASIS),
CONICET-UNR, Ocampo y Esmeralda, (S2000EZP) Rosario, Argentina. {\texttt{\{vallarella,haimovich\}@cifasis-conicet.gov.ar}}}

\begin{abstract}
{\color{black}
Discrete-time models of non-uniformly sampled nonlinear systems under zero-order hold relate the next state sample to the current state sample, (constant) input value, and sampling interval. The exact discrete-time model, that is, the discrete-time model whose state matches that of the continuous-time nonlinear system at the sampling instants may be difficult or even impossible to obtain. In this context, one approach to the analysis of stability is based on the use of an approximate discrete-time model and a bound on the mismatch between the exact and approximate models. This approach requires three conceptually different tasks: i) ensure the stability of the (approximate) discrete-time model,
ii) ensure that the stability of the approximate model carries over to the exact model, iii) if necessary, bound intersample behaviour. Existing conditions for ensuring the stability of a discrete-time model as per task i) have some or all of the following drawbacks: are only sufficient but not necessary; do not allow for varying sampling rate; cannot be applied in the presence of state-measurement or actuation errors. In this paper, we overcome these drawbacks by providing characterizations of, i.e.~necessary and sufficient conditions for, two stability properties: semiglobal asymptotic stability, robustly with respect to bounded disturbances, and semiglobal input-to-state stability, where the (disturbance) input may successfully represent state-measurement or actuation errors. Our results can be applied when sampling is not necessarily uniform.}
\end{abstract}

\begin{keyword}
  Sampled-data \sep
  nonlinear systems \sep
  non-uniform sampling \sep
  input-to-state stability (ISS) \sep
  {\color{black}discrete-time models}.
\end{keyword}
\end{frontmatter}

\section{Introduction}
\label{sec:introduction}

\begin{color}{black}

Most of the existing stability results for non-uniformly sampled systems deal with linear systems \cite{hetel2017,fujioka2009,suh2008}. Some results applicable to different classes of nonlinear systems were given in \cite{BERNUAU2017,Monaco2015,mattioni2017,karafyllis2009global}. In \cite{BERNUAU2017}, stabilization of homogeneous nonlinear systems with sampled-data inputs is analyzed in  by means of an emulation approach. In \cite{Monaco2015,mattioni2017}, multi-rate sampled-data stabilization via immersion and invariance for nonlinear systems in feedback form was developed. \cite{karafyllis2009global} presents sufficient conditions for uniform input-to-output and input-to-state stability for closed-loop systems with zero-order hold.

Strategies where sampling is inherently non-uniform and which have application to nonlinear systems are those of event- and self-triggered control.
In an event-triggered control strategy, the control action is computed based on the continuous-time 
system model (with the aid of a Lyapunov function, e.g.) and current state or output measurements, 
applied to the plant, and held constant until a condition that triggers the control action update
becomes true \cite{tabuad_tac07,Heemels2012}. The triggering condition requires continuous monitoring of some system variables,
and thus this type of event-triggered control does not exactly constitute
a sampled-data strategy. Other event-triggered strategies that verify the condition only periodically
have been developed for linear systems \cite{heemels2013,heemels2013model}. Self-triggered control \cite{Heemels2012,velfue_rtss03,Anta2010}, 
in addition to computing the current control action based on the continuous-time model, also computes the time instant at which the next control update will occur
requiring only sampled measurements.

Some approaches to stability analysis and control design for nonlinear systems under sampling and hold are based on the use of a discrete-time model for the sampled system. These approaches are referred to as Discrete-Time Design (DTD), or Sampled-Data Design (SDD) if, in addition, inter-sample behaviour is taken into account \cite{NesicSCL99,monaco2001issues}. As opposed to the linear-system case, the differential equations that describe a continuous-time nonlinear system's dynamics may be difficult or impossible to solve in closed form, and hence a discrete-time model exactly matching the state of the continuous-time system at the sampling instants is usually unavailable.
If the continuous-time system is input-affine then
the exact discrete-time model can be approximated to desired accuracy via the procedure in \cite{monaco2005differential,monaco2007advanced}.
Thus, DTD or SDD for nonlinear systems are usually based on an \emph{approximate} discrete-time model.

Interesting work on DTD for nonlinear systems under uniform sampling appears in \cite{NesicSCL99,NesicTAC04}. The results of \cite{NesicSCL99,NesicTAC04} are of the following conceptual form: given a specific bound on the mismatch between the 
exact and approximate discrete-time models (which can be known without having to compute the exact model) then some stability property on the approximate closed-loop model will carry over (in a practical sense) to the exact model for all sufficiently small sampling periods. These results have been extended to provide input-to-state and integral-input-to-state
stability results \cite{NesicLaila_TAC2002,NesicAngeliTAC2002}, to observer design \cite{Arcak20041931}, and to networked control systems \cite{van2012discrete}. All of these results are specifically suited to the case when sampling is uniform during operation or, in the case of \cite{van2012discrete}, when a nominal sampling period can be defined. Some extensions to the non-uniform sampling case were given in \cite{AADECA16,RPIC2017,valhai_arxiv18} that are also based on an approximate discrete-time model. Specifically, \cite{AADECA16} gives preliminary results to ensure the practical asymptotic stability of the exact discrete-time model under non-uniform sampling, \cite{valhai_arxiv18} gives a sufficient condition for the semiglobal practical input-to-state stability of the exact discrete-time model with respect to state-measurement errors, and \cite{RPIC2017} shows that a global 
stability property under uniform sampling, namely $(\beta,\R^n)$-stability, 
is equivalent to the analogous property under non-uniform sampling.

The aforementioned DTD approach requires two conceptually different tasks: i) ensure the stability of the (approximate) discrete-time model, and
ii) ensure that the stability of the approximate model carries over to the exact model. For the SDD approach, the following task should be added: iii) bound intersample behaviour. The existing conditions for ensuring the stability of a discrete-time model as per task i) have some or all of the following drawbacks: are only sufficient but not necessary; do not allow for varying sampling rate; cannot be applied in the presence of state-measurement or actuation errors. 

This paper addresses 
stability analysis for discrete-time models of sampled-data nonlinear
systems under the aforementioned DTD approach.
We characterize, i.e.~give necessary and sufficient conditions for, 
two stability properties: semiglobal asymptotic stability,
robustly with respect to bounded disturbances, 
and semiglobal input-to-state stability,
where the (disturbance) input may successfully represent state-measurement or actuation errors,
both specifically suited to non-uniform sampling.
In this context, the contribution of the current paper
is to overcome all of the drawbacks relating to task i) and mentioned in the previous paragraph. 
Our results thus apply to a discrete-time model of a sampled nonlinear system,
irrespective of how accurate this model may be. If the discrete-time model is only approximate, 
then our results can be used in conjunction with the results in \cite{valhai_arxiv18} in order to conclude
about the (practical) stability of the (unknown) exact model, as per task ii).

The motivation for the two stability properties characterized in the current paper comes in part from the fact that a discrete-time control law that globally stabilizes the exact discrete-time model under perfect state knowledge may cause some trajectories to be divergent under bounded state-measurement errors, as shown in \cite{valhai_arxiv18}. The difficulty in characterizing the robust semiglobal stability and semiglobal input-to-state stability properties considered (see Section~\ref{sec:preliminaries} for the precise definitions) is mainly due to their semiglobal nature and not so much to the fact that sampling may be non-uniform. The properties considered are semiglobal because the maximum sampling period for which stability holds may depend on how large the initial conditions are. This situation causes our derivations and proofs to become substantially more complicated than existing ones.
\end{color}

The organization of this paper is as follows. 
This section ends with a brief summary of the notation employed. 
In Section~\ref{sec:preliminaries} we state the problem and the required definitions and properties. 
Our main results are given in Section~\ref{sec:mainresults}.
An illustrative example is provided in Section~\ref{sec:examples} 
and concluding remarks are presented in Section~\ref{sec:conclusions}.
The appendix contains the proofs of some of the presented lemmas.

\textbf{Notation:}
$\R$, $\R_{\ge 0}$, $\N$ 
and $\N_0$ denote the sets of real, nonnegative real, natural and nonnegative integer numbers, respectively. 
We write $\alpha \in \K$ if $\alpha : \R_{\ge 0} \to \R_{\ge 0}$ is strictly increasing, 
continuous and $\alpha(0)=0$. We write $\alpha \in \Kinf$ if $\alpha\in\K$ and $\alpha$ is unbounded. 
We write $\beta \in \KL$ if $\beta : \R_{\ge 0} \times \R_{\ge 0} \to \R_{\ge 0}$, $\beta(\cdot,t)\in \K$ for all $t\ge 0$,
and $\beta(s,\cdot)$ is strictly decreasing asymptotically to $0$ for every $s$.
We denote the Euclidean norm of a vector $x \in \R^n$ by $|x|$. 
We denote an infinite sequence as $\{T_i\}:=\{T_i\}_{i=0}^{\infty}$. 
For any sequences $\{T_i\} \subset \R_{\ge 0}$ and $\{e_i\} \subset \R^m$, and any $\gamma\in\K$,
we take the following conventions: $\sum_{i=0}^{-1} T_i = 0$ and $\gamma(\sup_{0\le i\le -1}|e_i|) = 0$. 
Given a real number $T>0$ we denote by 
 \begin{color}{black}
$\Phi(T):=\{ \{T_i\} : \{T_i\} \text{ is such that }  T_i \in (0,T) \text{ for all } i\in \N_0 \}$
\end{color}
the set of all sequences of real numbers in the open interval $(0,T)$. 
For a given sequence we denote the norm $\|\{x_i\}\|:= \sup_{i\geq0} |x_i|$.

\section{Preliminaries}
\label{sec:preliminaries}

\begin{color}{black}

\subsection{Problem statement}
\label{sec:problem_statement}
\begin{color}{black}
We consider discrete-time systems that arise when modelling non-uniformly sampled continuous-time nonlinear systems of the form
\begin{equation}
\label{eq:cs}
\begin{split}
 \dot{x}&=f(x,u),\quad  x(0)=x_0, 
\end{split}
\end{equation}
under zero-order hold,
where $x(t) \in \R^n$, $u(t) \in \R^m$ are the state and control vectors respectively. 
We consider that the sampling instants 
$t_k$, $k\in \N_0$, satisfy $t_0 = 0$ and $t_{k+1} = t_k + T_k$, where $\{T_k\}_{k=0}^{\infty}$ is the sequence of corresponding sampling periods.
As opposed to the uniform sampling case where $T_k = T$ for all $k \in \N_0$,
we consider that the sampling periods may vary;
we refer to this situation as Varying Sampling Rate (VSR).
In addition, we assume that the current sampling period $T_k$ is known or determined at the current sampling instant $t_k$. This situation arises when the controller determines
the next sampling instant
according to a certain control strategy, such as in
self-triggered control;
we refer to this scheme as controller-driven sampling.
Due to zero-order hold, the continuous-time control signal $u$ is piecewise constant such that $u(t) = u(t_k) =: u_k$ for all $t\in [t_k,t_{k+1})$.
Given that the current sampling period $T_k$ is known or determined at the current sampling instant $t_k$, then the current control action $u_k$ may depend not only on the current state sample $x_k$ but also on $T_k$. If, in addition, state-measurement or actuation errors exist, then the true control action applied will also be affected by such errors. If we use $e_k$ to denote the considered error at the corresponding sampling instant, then we could have $u_k=U(x_k,e_k,T_k)$.

The class of discrete-time systems that arise when modelling a non-uniformly sampled continuous-time nonlinear system \eqref{eq:cs} under zero-order hold is thus of the form
\begin{align}
  \label{eq:fut}
  x_{k+1} &= F(x_k,u_k,T_k).
\end{align}
Our results apply to this class of discrete-time systems irrespective of whether the system model accurately describes the behaviour of some continuous-time system at the sampling instants or not. Of course, if the discrete-time model employed were the exact discrete-time model for some continuous-time system, then stability of the model could give some indication on the stability of the continuous-time system. Conditions on $f$ in (\ref{eq:cs}) for existence of the exact discrete-time model are given in \ref{app:dt-mod}. Since regrettably the exact discrete-time model is in general impossible to obtain, then approximate models should be used. Sufficient conditions for some stability properties to carry over from an approximate discrete-time model to the exact model were given in \cite{NesicTAC04} under uniform sampling and in \cite{valhai_arxiv18} for the controller-driven sampling case here considered. These conditions are based on bounds on the mismatch between the exact and approximate models and can be computed without having to compute the exact model.

As mentioned above, a control action $u_k$
computed from state measurements, having knowledge of the current sampling period and under the possible effect of state-measurement or actuation errors is of the form
\begin{equation}
  \label{eq:UxeT}
  u_k = U(x_k,e_k,T_k),
\end{equation}
where $e_k \in \R^q$ denotes the error and the dimension $q$ depends on the type of error. 
For example, if $e_k$ represents a state-measurement additive error, then $q=n$; 
if it represents actuation additive error, then $q=m$. 
Under (\ref{eq:UxeT}), the closed-loop model becomes
\begin{align}
  \label{eq:system1}
  x_{k+1} &=  F(x_k, U(x_k,e_k,T_k),T_k) =: \bar F(x_k,e_k,T_k)
\end{align}
which is once again on the form \eqref{eq:fut}.
We stress that a control law $u_k = \bar U(x_k,e_k)$
is also of the form (\ref{eq:UxeT}) and hence also covered by our results.
We will characterize two stability properties
for discrete-time models of the form \eqref{eq:system1}: 
robust semiglobal stability and semiglobal input-to-state stability.
For the sake of clarity of the proofs
we will use $d_k \in \R^p$ instead of $e_k$ to represent bounded disturbances that do not destroy asymptotic stability.
Given $D>0$, we define $\D:=\big\{\{d_i\} \subset \R^p : |d_i| \leq D, \forall i\in\N_0\big\}$, the set of all disturbance sequences whose norm is not greater than $D$. Thus, for our robust stability results, we will consider a discrete-time model of the form
\begin{align}
 \label{eq:sysd}
  x_{k+1} &=  \bar F(x_k,d_k,T_k), \quad \{d_i\}\in  \D.
\end{align}

\end{color}
\end{color}

\subsection{Stability properties for varying sampling rate}
\label{sec:preliminaries2}

\begin{color}{black}
The next definitions are extensions of stability properties in
\cite{NesicSCL99,sontag_tac89,Sontag95,Sontag952,NesicLaila_TAC2002}.
The first one can be seen as a robust and semiglobal (with respect to initial states) version of $(\beta,\R^n)$-stability of \cite{NesicSCL99}, suitable for the non-uniform sampling case. The second definition presents the discrete-time 
global, semiglobal and semiglobal practical versions of the input-to-state stability (ISS) for non-uniform sampling.
\end{color}
\begin{color}{black}
\begin{definition}
  \label{def:ss-vsr}
  The system \eqref{eq:sysd} is said to be \textit{Robustly Semiglobally Stable under Varying Sampling Rate} (RSS-VSR)
  if there exists a function $\beta \in \KL$ such that for every $M \geq0$ there exists $T^\blacktriangle=T^\blacktriangle(M)>0$ such that the solutions of 
  \eqref{eq:sysd}
  satisfy
  \begin{align}
    \label{eq:semiglobal_cond}
    |x_k|\leq \beta\left(|x_0|,\sum_{i=0}^{k-1} T_i \right)
  \end{align}
  for all\footnote{As explained under ``Notation'' in Section~\ref{sec:introduction}, for $k=0$ we interpret $\sum_{i=0}^{-1} T_i = 0$ 
  and $\gamma(\sup_{0\le i \le -1} |e_i| )=0$.} $k \in \N_0$, $\{T_i\} \in \Phi(T^\blacktriangle)$, $|x_0|\leq M$ and $\{d_i\} \in \D$.
\end{definition}

\begin{remark}
  \label{rem:Ttriangle}
  Without loss of generality, the function $T^\blacktriangle(\cdot)$ in Definition~\ref{def:ss-vsr} can be taken nonincreasing.
\end{remark}

The RSS-VSR property is semiglobal because the bound $T^\blacktriangle$ on the sampling periods may depend on how far from the origin the initial conditions may be (as quantified by $M$). 
If there exists $\beta \in \KL$ and $T^\blacktriangle>0$
such that \eqref{eq:semiglobal_cond} holds
for all $k \in \N_0$, $\{T_i\} \in \Phi(T^\blacktriangle)$ and $x_0 \in  \R^n$,
then the system is said to be globally Robustly Stable under VSR (RS-VSR). 
When disturbances are not present ($\mathcal{D} = \{0\}$, i.e. $D=0$), the RS-VSR property becomes $(\beta,\R^n)$-stability under VSR \cite{RPIC2017}. If in addition to lack of disturbances, uniform sampling is imposed ($T_k = T$ for all $k\in\N_0$), then RS-VSR becomes $(\beta,\R^n)$-stability \cite{NesicSCL99}. In \cite{RPIC2017}, it was shown that existence of $\beta\in \KL$ such that a system is $(\beta,\R^n)$-stable is equivalent to existence of $\tilde\beta\in\KL$ such that it is $(\tilde\beta,\R^n)$-stable under VSR.
\end{color}

\begin{definition} 
  \label{def:ISS}
  The system (\ref{eq:system1}) is said to be
  \begin{enumerate}
   \item {\color{black}ISS-VSR if there exist functions $\beta \in \mathcal{KL}$ and $\gamma \in \K_\infty$ 
   and a constant $T^\star>0$ such that the solutions of \eqref{eq:system1} satisfy
      \begin{align}
      \label{eq:S-ISS-VSR}
      |x_k|\leq \beta\left(|x_0|,\sum_{i=0}^{k-1}T_i\right)+\gamma\left(\sup_{0\le i \le k-1}|e_i| \right ),
    \end{align}
for all $k\in \N_0$, $\{T_i\} \in \Phi(T^\star)$, $x_0\in \R^n$ and $\{e_i\} \subset \R^p$.}
   \item \textit{Semiglobally ISS-VSR} (S-ISS-VSR) if there exist functions $\beta \in \mathcal{KL}$ and $\gamma \in \K_\infty$
    such that for every $M\geq0$ and $E\geq0$ there exists
    $T^\star=T^\star(M,E)>0$ such that the solutions of \eqref{eq:system1} satisfy \eqref{eq:S-ISS-VSR}
    for all $k\in \N_0$, $\{T_i\} \in \Phi(T^\star)$, $|x_0|\leq M$ and $\|\{e_i\}\|\leq E$.
    \item     \textit{Semiglobally Practically ISS-VSR} (SP-ISS-VSR) 
    if there exist functions $\beta \in \mathcal{KL}$ and 
    $\gamma \in \K_\infty$ such that for every $M\geq0$, $E\geq0$ and 
    $R>0$ there exists $T^\star=T^\star(M,E,R)>0$ such that the solutions of \eqref{eq:system1} satisfy
        \begin{align}
      \label{eq:SP-ISS-VSR}
      |x_k|\leq \beta\left(|x_0|,\sum_{i=0}^{k-1}T_i\right)+\gamma\left(\sup_{0\le i \le k-1}|e_i| \right )+R,
    \end{align}
     for all $k\in \N_0$, $\{T_i\} \in \Phi(T^\star)$, $|x_0|\leq M$ and $\|\{e_i\}\|\leq E$.
  \end{enumerate}
\end{definition}

Note that 
$\text{ISS-VSR } \Rightarrow \text{ S-ISS-VSR } \Rightarrow \text{ SP-ISS-VSR}$.

\begin{color}{black}
\begin{remark}
  \label{rem:Ttrnoninc}
  Without loss of generality, 
  the function $T^\star(\cdot,\cdot)$ in the definition 
  of S-ISS-VSR in Definition~\ref{def:ISS} 
  can be taken nonincreasing in each variable.
\end{remark}
\end{color}

\color{black}

\section{Main Results}
\label{sec:mainresults}
{\color{black}
In this section, we present characterizations of the RSS-VSR and S-ISS-VSR properties defined in Section~\ref{sec:preliminaries2}. In Lemma~\ref{lemma:equiva}, $\epsilon$-$\delta$ and Lyapunov-type characterizations
are given for the RSS-VSR property. The main difference between these characterizations and the existing characterizations of $(\beta,\R^n)$-stability \cite[Lemma~4]{NesicSCL99} and $(\beta,\R^n)$-stability under VSR \cite[Lemma~2]{RPIC2017} depend on the semiglobal nature of RSS-VSR. The proof of Lemma~\ref{lemma:equiva} is given in~\ref{A}.}
%
\begin{lemma}
  \label{lemma:equiva}
  The following statements are equivalent:
  \begin{enumerate}
  \item The system \eqref{eq:sysd} is RSS-VSR. \label{it:RSSVSR}
  \item For every $M \geq 0$ there exists $T^\blacktriangle=T^\blacktriangle(M)>0$ so that\label{item:eps-delt}
   \begin{enumerate}[i)]
    \item \label{enum:lemequiv1} for all $\epsilon>0$, there exists $\delta = \delta(\epsilon) > 0$ ($\delta$ is independent of $M$) such that the solutions of (\ref{eq:sysd}) with $|x_0|\leq \min \{\delta,M\}$, $\{T_i\} \in \Phi(T^\blacktriangle)$ and $\{d_i\} \in \D$ satisfy $|x_k|\leq \epsilon$ for all $k\in \N_0$,    
    \item \label{enum:lemequiv3} for all $L\ge 0$, there exists $C=C(M,L)\ge 0$ such that the solutions of (\ref{eq:sysd}) with $|x_0|\leq M$, $\{T_i\} \in \Phi(T^\blacktriangle)$ and $\{d_i\} \in \D$ satisfy $|x_k|\leq C$, for all $k \in \N_0$ for which $\sum_{i=0}^{k-1}T_i\le L$, and
    \item \label{enum:lemequiv2} for all $\epsilon>0$, there exists $\mathcal{T} = \mathcal{T}(M,\epsilon) \ge 0$ 
    such that the solutions of (\ref{eq:sysd}) with $|x_0|\leq M$, $\{T_i\} \in \Phi(T^\blacktriangle)$ and $\{d_i\} \in \D$ satisfy $|x_k|\leq \epsilon$, for all $k \in \N_0$ for which $\sum_{i=0}^{k-1}T_i\geq \mathcal{T}$.
    \end{enumerate}
     \item \label{it:boundedV} There exist $\alpha_1, \alpha_2, \alpha_3 \in \K_{\infty}$ such that for every $M\geq0$ there exist $T^*= T^*(M)>0$ and $V_M: \R^n \rightarrow \R_{\geq0} \cup \{\infty\}$ such that 
    \begin{subequations}
      \label{eq:cond_1_semiglobal_32}
      \begin{align}
        \label{eq:cond_1_semiglobal_a32}
        & \alpha_1(|x|)\leq V_M(x),\quad \forall x \in \R^n,\\
        \label{eq:cond_1_semiglobal_b32}
        & V_M(x) \le \alpha_2(|x|),\quad \forall |x|\le M,
      \end{align}
    \end{subequations}
    and
    \begin{align}
      \label{eq:cond_2_semiglobal_32}
      &V_M(\bar F(x,d,T))-V_M(x) \leq -T\alpha_3 (|x|) \quad  
    \end{align}
    for all  $|x|\leq M$, $|d|\leq D$ and $T\in(0,T^*)$.
  \end{enumerate}
\end{lemma}

{\color{black}
The $\epsilon$-$\delta$ characterization in item~\ref{item:eps-delt}.~of
Lemma~\ref{lemma:equiva} contains all the
ingredients 
of an 
$\epsilon$-$\delta$ characterization of uniform global asymptotic 
stability for a continuous-time system \cite{masser_am56} but in semiglobal
form and for a discrete-time model. These ingredients are: semiglobal 
uniform stability in \ref{enum:lemequiv1}), semiglobal uniform boundedness 
in \ref{enum:lemequiv3}), and semiglobal uniform attractivity in
\ref{enum:lemequiv2}). 
The Lyapunov conditions in item~\ref{it:boundedV}.~have several differences 
with respect to the Lyapunov-type conditions ensuring $(\beta,\R^n)$-stability \cite{NesicSCL99} or $(\beta,\R^n)$-stability under VSR \cite{RPIC2017}. 
First, note that the Lyapunov-type
function $V_M$ may be not the same for each upper bound $M$ on the norm
of the state. Second, the functions $V_M$ may take infinite values
and it is not required that they satisfy any Lipschitz-type condition.
Third, the upper bound given by $\alpha_2\in\Ki$ should only hold for 
states whose norm is upper bounded by $M$.
}

Theorem~\ref{lem:ISS-VSR=LF} gives necessary and sufficient conditions 
for a discrete-time model of the form \eqref{eq:system1} to be S-ISS-VSR.
These conditions consist of specific boundedness and continuity requirements
and a Lyapunov-type condition.
The characterizations given in Lemma~\ref{lemma:equiva}
are used in the proof of  
Theorem~\ref{lem:ISS-VSR=LF}.
\begin{theorem}
  \label{lem:ISS-VSR=LF}
  The following statements are equivalent:
   \begin{enumerate}
  \item \label{it:ISS-VSR=LF1}
    The system  \eqref{eq:system1}  is S-ISS-VSR.
  \item 
  \begin{enumerate}[i)]
  \item  There exists $\mathring T>0$ so that $\bar F(0,0,T) = 0$ for all $T\in(0,\mathring T)$. \label{item:zero}
\item There exists $\hat T>0$ such that for every $\epsilon>0$ there exists $\delta=\delta(\epsilon)>0$ 
such that $|\bar F(x,e,T)|<\epsilon$ whenever $|x|\leq \delta$, $|e|\leq \delta$ and $T\in(0,\hat T)$.\label{item:continuity}
\item For every $M\ge 0$ and $E\ge 0$, there exist $C=C(M,E)>0$ and $\check{T} = \check{T}(M,E) > 0$, with $C(\cdot,\cdot)$ 
nondecreasing in each variable and $\check{T}(\cdot,\cdot)$ nonincreasing in each variable, such that $|\bar F(x,e,T)|\leq C$ 
for all $|x|\leq M$, $|e|\leq E$ and $T \in (0,\check{T})$.   \label{item:bound}
\item There exist $\alpha_1, \alpha_2, \alpha_3 \in \K_{\infty}$ 
and $\rho \in \K$ such that for every $M\ge0$ and $E\ge0$ there exist 
$\tilde T = \tilde T(M,E) > 0$ and $V=V_{M,E}: \R^n \rightarrow \R_{\geq0} \cup \{\infty\}$ such that 
\label{item:lyapunov}
    \begin{subequations}
      \label{eq:cond_1_semiglobal}
      \begin{align}
        \label{eq:cond_1_semiglobal_a}
        & \alpha_1(|x|)\leq V(x),\quad \forall x \in \R^n,\\
        \label{eq:cond_1_semiglobal_b}
        & V(x) \le \alpha_2(|x|),\quad \forall |x|\le M,
      \end{align}
    \end{subequations}
    and
  \begin{align}
    \label{eq:lypunov_c2_SISSVSR} 
    &V(\bar F(x,e,T)) -V(x) \leq -T\alpha_3 (|x|)
  \end{align} 
  for all $\rho(|e|)\leq |x| \leq M$, $|e|\le E$ and $T\in (0,\tilde T)$.
\end{enumerate}
\end{enumerate}
\end{theorem}

\begin{proof}[Proof of Theorem~\ref{lem:ISS-VSR=LF}] 

(1 $\Rightarrow$ 2)
Let $\beta_0 \in \mathcal{KL}$, $\gamma_0 \in \K_\infty$ and $T^\star(\cdot,\cdot)$ 
characterize the S-ISS-VSR property.

1) $\Rightarrow$ \ref{item:zero}) 
Define $\mathring T:=T^\star(0,0)$. From (\ref{eq:system1}) and \eqref{eq:S-ISS-VSR}, we have 
\begin{align}
  |\bar F(0,0,T)| &\leq \beta_0(0,T)+\gamma_0(0) = 0 \notag 
\end{align}
for all $T\in(0, \mathring T)$. 

1) $\Rightarrow$ \ref{item:continuity})
Define $\hat\beta,\eta \in \K_\infty$ via $\hat \beta(s):= \beta_0(s,0)+s$ and $\eta(s) := \min\{ \hat \beta^{-1}(s/2), \gamma_0^{-1}(s/2)\}$. Define $\hat T := T^\star(\eta(1),\eta(1))$.
Let $\epsilon>0$. Choose $\delta = \eta(\min\{\epsilon,1\}) > 0$. Note that $T^\star(\delta,\delta) \ge \hat T$ because $\delta \le \eta(1)$ and $T^\star$ is nonincreasing in each variable. Then, using (\ref{eq:system1}) and \eqref{eq:S-ISS-VSR}, it follows that for all $|x|\leq \delta$, $|e|\le \delta$ and $T\in(0,\hat T)$ we have
\begin{color}{black}
\begin{align} 
|\bar F(x,e,T)|&\leq \beta_0(\delta,T)+\gamma_0(\delta) 
                       <\hat \beta(\delta)+\gamma_0(\delta) \leq \epsilon.
\end{align}\end{color}
1) $\Rightarrow$ \ref{item:bound})
Let $\check T = T^\star$ and $C(M,E) = \beta_0(M,0) + \gamma_0(E)$. Then, $\check T$ is nonincreasing in each variable and $C$ is increasing (and hence nondecreasing) in each variable. Let $M,E\ge 0$. Then, from (\ref{eq:system1}) and \eqref{eq:S-ISS-VSR}, for all $|x|\leq M$, $|e|\leq E$ and $T\in (0,\check{T}(M,E))$ we have that $|\bar F(x,e,T)|\leq \beta_0\left(M,0\right)+\gamma_0\left(E \right )=C(M,E)$. 

1) $\Rightarrow$ \ref{item:lyapunov}) 
Define $\beta(s,t):=2\beta_0(s,t)$ and $\gamma(s):=2\gamma_0(s)$. 
Define $\alpha \in \K_\infty$ via  $\alpha(s) := \beta(s,0)$ and $\sigma\in \K_\infty$ via $\sigma(s):=\gamma^{-1}(\frac{1}{2}\alpha^{-1}(s))$. Consider the following system: 
\begin{align}
  \label{eq:simp}
  x_{k+1}=\bar F(x_k,\sigma(|x_k|)d_k,T_k), \quad \|\{d_i\}\|\leq 1, 
\end{align}
with $d_k \in \R^q$ for all $k\in\N_0$.

\begin{claim}
  \label{clm:1}
  For every $M\ge 0$ there exists $\bar T=\bar T(M)>0$, 
  with  $\bar T(\cdot)$  nonincreasing, 
  such that the
solutions of \eqref{eq:simp} satisfy 
\begin{align}
  \label{eq:bndclm1}
  |x_k| &\leq \max \left\{ \beta\left(|x_0|, \sum_{i=0}^{k} T_i  \right) , \frac{1}{2}|x_0| \right \} \le \alpha(|x_0|) 
\end{align}
for all $k\in \N_0$, whenever $|x_0|\leq M$ and $\{T_i\}\in \Phi(\bar T)$.
\end{claim}
\noindent\emph{Proof of Claim~\ref{clm:1}:}
Given $M\ge 0$, take $\bar T(M) = T^\star(M,\gamma^{-1}(M/2)) > 0$. 
Note that $\bar T$ is nonincreasing because $\gamma^{-1}\in \K_\infty$ and $T^\star$ is nonincreasing in each variable.
We establish the result by induction. For $k=0$, we have $|x_0| \le \beta_0(|x_0|,0) = \alpha(|x_0|)$. Suppose that $|x_i| \le \alpha(|x_0|)$ for all $0 \le i \le k$. 
Then, $|\sigma(|x_i|)d_i| \le \sigma(|x_i|) \le \sigma(\alpha(|x_0|)) = \gamma^{-1}(|x_0|/2) \le \gamma^{-1}(M/2)$ for all $0 \le i \le k$. 
Then, for all $\{T_i\}\in \Phi(\bar T(M))$, we have
\begin{align}
|x_{k+1}|  &\le \beta_0\left(|x_0|,\sum_{i=0}^{k} T_i \right) + \gamma_0\left( \sup_{0\le i\le k} \big| \sigma(|x_i|)d_i \big| \right)\\
           &\leq \max \left\{ \beta\left(|x_0|, \sum_{i=0}^{k} T_i  \right) , \gamma\left(\sup_{0\le i \le k}|\sigma(|x_i|)d_i| \right)\right \} \notag \\
           &\leq \max \left\{ \beta\left(|x_0|, \sum_{i=0}^{k} T_i  \right) , \frac{1}{2}|x_0| \right \} 
               \le \alpha(|x_0|), \label{eq:induc}
\end{align}
where in the last inequality we have used the fact
that $2s\leq 2\beta_0(s,0)=\alpha(s)$.
By induction, then $|x_k|\leq \alpha(|x_0|)$ for all $k\in \N_0$.\mer

We next show that \eqref{eq:simp} is RSS-VSR by means of Lemma~\ref{lemma:equiva}. 
Given $M\ge 0$, take $T^\blacktriangle(M) = \bar T(M)$.  

Condition \ref{enum:lemequiv1}) of Lemma~\ref{lemma:equiva}: 
Let $\epsilon>0$ and take $\delta = \alpha^{-1}(\epsilon)$. Then, if $|x_0| \le \min\{\delta,M\}$ and $\{T_i\} \in \Phi(T^\blacktriangle(M))$, by Claim~\ref{clm:1} it follows that $|x_k|\leq \alpha(|x_0|) \le \epsilon$ for all $k\in \N_0$.

Condition \ref{enum:lemequiv3}) of Lemma~\ref{lemma:equiva}: 
Define $C(M,L) = \alpha(M)$. Then, if $|x_0| \le M$ and $\{T_i\} \in \Phi(T^\blacktriangle(M))$, by Claim~\ref{clm:1} it follows that $|x_k|\le \alpha(|x_0|) \le C$ for all $k\in \N_0$.

Condition \ref{enum:lemequiv2})  of Lemma~\ref{lemma:equiva}:
For every $j\in\N_0$, define $M_j:=\frac{M}{2^j}$, $t_{j}=t_j(M)>0$ via 
\begin{align*}
  \beta\left(M_j,t_j\right) = \frac{1}{2}M_j,
\end{align*}
and $\tau_j = \tau_j(M)$ via
\begin{align*}
  \tau_j(M) &:= jT^\blacktriangle(M) + \sum_{i=0}^{j} t_i(M),  
\end{align*}
\begin{claim}
  \label{clm:2h}
  Consider $|x_0|\leq M$ and $\{T_i\} \in \Phi(T^\blacktriangle(M))$.
  For all $j,k\in\N_0$ for which $\sum_{i=0}^{k-1} T_i \ge \tau_j$, it happens that
  \begin{align*}
    |x_k| &\le M_{j+1}.
  \end{align*}
\end{claim}
\noindent\emph{Proof of Claim~\ref{clm:2h}:} By induction on $j$. 
If $\sum_{i=0}^{k-1} T_i \ge \tau_0 = t_0$, 
then from (\ref{eq:bndclm1}) in Claim~\ref{clm:1} and since $M_0=M$, we have that
\begin{align*}
  |x_k| &\le \max\{\beta(M_0,t_0),M_0/2\} = M_1,
\end{align*}
Hence, our induction hypothesis holds for $j=0$.
Next, suppose that for some $j\in\N_0$ and for all $k\in\N_0$ for which $\sum_{i=0}^{k-1} T_i \ge \tau_j$ it happens that $|x_k| \le M_{j+1}$. 
Let $k^* = \min\{k\in\N_0 : \sum_{i=0}^{k-1} T_i \ge \tau_j \}$. Then, $\sum_{i=0}^{k^* - 1} T_i < \tau_j + T^\blacktriangle$ and $|x_{k^*}| \le M_{j+1} \le M_0 = M$. If $\sum_{i=0}^{k-1} T_i \ge \tau_{j+1} = \tau_j + T^\blacktriangle + t_{j+1}$, then necessarily $\sum_{i=k^*}^{k-1} T_i \ge t_{j+1}$. 
Note that  $\Phi(T^\blacktriangle(M_0))\subset \Phi(T^\blacktriangle(M_j))$ for all $j\in \N_0$.
By Claim~\ref{clm:1} and time invariance, it follows that, for all 
$\{T_i\} \in \Phi(T^\blacktriangle(M_0))$ and all
$k$ for which $\sum_{i=0}^{k-1} T_i \ge \tau_{j+1}$, then
\begin{align*}
  |x_k| &\le \max\left\{ \beta\left( \left|x_{k^*} \right|, \sum_{i=k^*}^{k-1} T_i \right) , \frac{1}{2} \left|x_{k^*} \right| \right\}\\
  &\le \max\left\{ \beta\left( M_{j+1}, t_{j+1} \right) , \frac{1}{2} M_{j+1} \right\} = M_{j+2}.
\end{align*}
Therefore, our induction hypothesis holds for $j+1$.\mer

Given $\epsilon>0$ define $p = p(M,\epsilon) \in\N$ and $\T(M,\epsilon)$ as
\begin{align*}
  p(M,\epsilon) &:= \min \{ j \in \N_0 : M_j < 2\epsilon \},\\
  \T(M,\epsilon) &:= \tau_p(M,\epsilon).
\end{align*}
By Claim~\ref{clm:2h}, it follows that 
for all $|x_0|\leq M$,
all $\{T_i\}\in \Phi(T^\blacktriangle(M))$
and all $k\in\N_0$ for which $\sum_{i=0}^{k-1} \ge \T(M,\epsilon)$, then $|x_k| \le M_{p+1} < \epsilon$.
Therefore, by Lemma~\ref{lemma:equiva}, the system \eqref{eq:simp} is RSS-VSR
and there exist $\alpha_1, \alpha_2, \alpha_3
\in \K_{\infty}$ 
such that for every $M\geq 0$ there exist
$T^*=T^*(M)>0$ and $V_M: \R^n \rightarrow \R_{\geq0} \cup \{\infty\}$ such that 
\eqref{eq:cond_1_semiglobal_32} holds.
Consider $E\ge0$ given and
define  $\tilde T(M,E):=T^*(M)$,
then \eqref{eq:cond_1_semiglobal} holds. 
Also 
\begin{align}
\label{eq:19}
V_M(&\bar F(x,\sigma(|x|)d,T))-V_M(x) \leq -T\alpha_3 (|x|), 
\end{align}
holds for all $|x|\leq M$, all $|d|\leq 1$ and all $T\in(0, \tilde T)$.
Select $\rho(s):=\sigma^{-1}(s)$.
Then, for all $|e|\leq E$ such that $\rho(|e|) \leq |x|$
we have $|e|\leq \sigma(|x|)$.
Therefore all $e$ such that $\rho(|e|) \leq |x|$
can always be written as $e=\sigma(|x|)d$ for some $d\in \R^q$
with $|d|\leq1$. Then, from \eqref{eq:19}, we have that
\eqref{eq:lypunov_c2_SISSVSR} holds.

(2 $\Rightarrow$ 1) 
We aim to prove that there exist 
$\beta \in \KL$ and $\gamma \in \K_\infty$ such that for all $M_0\ge0$, $E_0\ge0$ 
there exists $T^\star(M_0,E_0)>0$ such that 
the solutions of \eqref{eq:system1} satisfy
\begin{align}
|x_k|\leq \beta\left(|x_0|,\sum_{i=0}^{k-1}T_i\right)+\gamma\left(\sup_{0\le i \le k-1}|e_i| \right )
\end{align}
for all $k\in \N_0$, all $\{T_i\}\in \Phi(T^\star)$,
all $|x_0|\leq M_0$, and all $\|\{e_i\}\|\leq E_0$.
Consider $\rho\in\K$ from \ref{item:lyapunov}). Define, $\forall s \geq0$, 
\begin{align}
  \X_1(s)&:= \{x\in\R^n: |x|\leq \rho(s) \} \label{eq:X1_SISSVSR} \\
  \E(s)&:= \{e\in\R^q: |e|\leq s \}  \label{eq:E_SISSVSR}   \\
 \bar T(s) &:= \min\Big\{\mathring T, \hat T, \check{T}(\rho(s),s)  \Big\} \label{eq:barT_SISSVSR}\\
  \S(s) &:= \X_1(s)\times\E(s)\times(0,\bar T(s))  \label{eq:setS_SISSVSR} \\
  \sigma(s) &:=\sup_{(x,e,T) \in \S(s)} |\bar F(x,e,T)|. \label{eq:1111_SISSVSR} 
\end{align}

\begin{claim}
\label{clm:6}
There exists $\zeta \in \K_\infty$ such that $\zeta \geq \sigma$.
\end{claim}
\indent\emph{Proof of Claim~\ref{clm:6}:}
From \eqref{eq:X1_SISSVSR}--\eqref{eq:E_SISSVSR}, we have $\X_1(0)=\{0\}$ and $\E(0)=\{0\}$. 
From assumptions~\ref{item:zero})--\ref{item:bound}), 
then $\bar T(0) > 0$ and $\sigma(0)=0$. 
We next prove that $\sigma$ is right-continuous at zero. 
Let $\epsilon > 0$ and take $\delta=\delta(\epsilon)$ according to \ref{item:continuity}). 
Define $\hat \delta:=\min\left\{\delta,\rho^{-1}(\delta)\right\}$
(if $\delta \notin \text{dom} \ms \rho^{-1}$, just take $\hat\delta=\delta$).
Then for all $x\in \X_1(\hat \delta)$ and $e\in \E(\hat \delta)$ the inequalities $|x|\leq \delta$ and $|e|\leq \delta$ hold. 
Consequently, by \ref{item:continuity}), we have $\sigma(s)\leq \epsilon$  for all $0\leq s\leq \hat \delta$.
This shows that $\lim_{s\to 0^+} \sigma(s) = \sigma(0) = 0$.

From \ref{item:bound}), 
it follows that $|\bar F(x,e,T)| \le C(M,E)$ for all $|x|\le M$, $|e|\le E$ and
$T\in (0,\check{T}(M,E))$. 
From (\ref{eq:X1_SISSVSR})--(\ref{eq:1111_SISSVSR}) and the fact that
$C(\cdot,\cdot)$ is nondecreasing in each variable, it follows that
$\sigma(s) \le C(\rho(s),s)$ for all $s\geq 0$.
Then, we have $\sigma : \R_{\geq0} \rightarrow \R_{\geq0}$, $\sigma(0)=0$, $\sigma$
is right-continuous at zero and
bounded by a nondecreasing function.
By \cite[Lemma~2.5]{clarke1998asymptotic}, 
there exists a function $\zeta \in \K_\infty$ 
such that $\zeta \geq \sigma$. \mer

Define $\eta \in \K_\infty$ via
\begin{align}
  \label{eq:defEta}
  &\eta(s) := \max \{ \zeta (s),\rho(s) \}\quad \forall s\ge 0. 
\end{align}

Consider $M_0\ge0$ and $E_0\ge0$ given and $\alpha_1,\alpha_2,\alpha_3 \in \Kinf$ from \ref{item:lyapunov}). 
Select $E:=E_0$
and 
\begin{align}
\label{eq:defM}
M:= \alpha_1^{-1} \comp \alpha_2\big( \max\{M_0, \eta(E_0)\} \big).
\end{align}
Let \ref{item:lyapunov}) generate $\tilde T=\tilde T(M,E)>0$ and 
$V_{M,E}:\R^n \rightarrow \R_{\geq0} \cup \{\infty\}$ such 
that 
\eqref{eq:cond_1_semiglobal}
and (\ref{eq:lypunov_c2_SISSVSR}) hold. 
Note that $M\geq  \max \{ M_0, \eta(E_0)\}$
and that $\tilde T(M,E)\leq \tilde T(M_0,E_0)$ because $\tilde T$ is nonincreasing in each variable.
Define $T^\star=\min\{\tilde T, \bar T(E)\}$ and
\begin{align}
  \X_2(s)&:= \{x: V_{M,E}(x)\leq \alpha_2(\eta(s))\}. \label{eq:X2}
\end{align}
Consider that $0\leq s\leq \eta^{-1}(M)$.
Let $x\in\X_1(s)$, by \eqref{eq:X1_SISSVSR} and \eqref{eq:defEta}, we have $|x| \le \rho(s) \le \eta(s) \leq M$
for all $0\leq s\leq \eta^{-1}(M)$. Then, 
by \eqref{eq:cond_1_semiglobal_b},
$V_{M,E}(x) \leq \alpha_2(|x|) \leq \alpha_2(\eta(s))$ for all $0\leq s\leq \eta^{-1}(M)$.
Therefore, $\X_1(s)\subseteq \X_2(s)$ 
for all $0 \le  s \leq \eta^{-1}(M)$.
Let $x_k:=x(k,x_0,\{e_i\},\{T_i\})$ denote the solution to \eqref{eq:system1} corresponding to $|x_0|\leq M_0$, $\|\{e_i\}\| \le E_0$ and $\{T_i\}\in \Phi(T^\star)$.
From (\ref{eq:cond_1_semiglobal_b}) if $x_k$ satisfies
$|x_k|\leq M$ we have $\alpha_2^{-1}(V_{M,E}(x_{k}))\leq |x_k|$; 
using this in (\ref{eq:lypunov_c2_SISSVSR}) then
\begin{align}
  \label{eq:diffVkp2} 
  V_{M,E}(x_{k+1}) &-V_{M,E}(x_{k}) \leq -T_k\alpha_3 (|x_k|) \le -T_k \alpha(V_{M,E}(x_{k})) \notag \\
             & \quad \text{if} \ms\ms \rho(|e_k|)\leq |x_k|\leq M
\end{align}
where $\alpha \dfn \alpha_3 \comp \alpha_2^{-1}$. 

\begin{claim}
\label{clm:3}
If $|x_0| \le M_0$ then $|x_k| \le M$ for all $k\in\N_0$.
\end{claim}
\indent\emph{Proof of Claim~\ref{clm:3}:}
By induction we will prove that $V_{M,E}(x_k) \le \alpha_2( \max\{ M_0, \eta(E_0) \})$ for all $k\in\N_0$. 
For $k=0$, from \eqref{eq:cond_1_semiglobal},
we have $|x_0| \le M_0 \leq M$ implies that $V_{M,E}(x_0) \le \alpha_2(M_0) \le \alpha_2( \max\{ M_0, \eta(E_0) \})$.
Suppose that for some $k\in \N_0$ it happens that
$V(x_k) \le \alpha_2( \max\{ M_0,\eta(E_0) \})$. 
Then, by \eqref{eq:cond_1_semiglobal_a} and \eqref{eq:defM}, $|x_k| \le M$. 
If $x_k \notin \X_1(|e_k|)$, then $|x_k| >  \rho(|e_k|)$ and,
from (\ref{eq:diffVkp2}), then $V(x_{k+1}) \le V(x_k)$. If $x_k \in \X_1(|e_k|)$, from \eqref{eq:1111_SISSVSR},
the definition
of $\eta$ and \eqref{eq:defM} we have $|x_{k+1}| \leq \eta(|e_k|)\leq \eta(E_0) \leq \alpha^{-1} \comp \alpha_2 \comp \eta(E_0) \leq M$.
Using (\ref{eq:cond_1_semiglobal_b}), then $V(x_{k+1}) \le \alpha_2\comp\eta(|e_k|) \le \alpha_2\comp\eta(E_0)$,
and hence the induction assumption holds for $k+1$.
Since $V(x_k) \le \alpha_2( \max\{ M_0, \eta(E_0) \})$ 
implies that $|x_k| \le M$, we have thus shown that $|x_k| \le M$ for all $k\in\N_0$.\mer

\begin{claim}
\label{clm:4}
Consider $\|\{e_i\}\| \leq E_0$.
If $x_{\ell}\in \X_2(\|\{e_i\}\|)$ for some $\ell\in \N_0$ then $x_k$ remains in $\X_2(\|\{e_i\}\|)$ for all $k\geq \ell$.
\end{claim}
\indent\emph{Proof of Claim~\ref{clm:4}:}
By definition of $M$ in \eqref{eq:defM} we have
$E_0\leq \eta^{-1} \comp \alpha_2^{-1} \comp \alpha_1 (M)$.
Let $x_{\ell} \in \X_2(\|\{e_i\}\|)$,
then $V_{M,E}(x_\ell) \leq \alpha_2 \comp \eta(\|\{e_i\}\|)$.
By \eqref{eq:cond_1_semiglobal_a}, 
$|x_\ell| \leq \alpha_1^{-1} \comp \alpha_2 \comp \eta(\|\{e_i\}\|) \leq  \alpha_1^{-1} \comp \alpha_2 \comp \eta(E_0) \leq M $.
If $x_{\ell} \notin \X_1(\|\{e_i\}\|)$, then $|x_\ell| >  \rho(\|\{e_i\}\|)$. Consequently, if $x_{\ell} \in \X_2(\|\{e_i\}\|) \setminus \X_1(\|\{e_i\}\|)$, from \eqref{eq:diffVkp2} it follows that
\begin{align}
V_{M,E}(x_{\ell+1}) &\leq V_{M,E}(x_{\ell})-T_{\ell} \alpha(V_{M,E}(x_{\ell}))  \notag 
 \leq V_{M,E}(x_{\ell})  
\end{align}
and hence $x_{\ell+1} \in \X_2(\|\{e_i\}\|)$. 
Next, consider that  $x_{\ell} \in \X_1(\|\{e_i\}\|)$.
From \eqref{eq:1111_SISSVSR}, Claim~\ref{clm:6}, \eqref{eq:defEta}
and the definition of $M$ 
we have $|x_{\ell+1}| \leq \eta(\|\{e_i\}\|) \leq \eta(E_0) \leq M$.
Using (\ref{eq:cond_1_semiglobal_b}) and recalling \eqref{eq:X2}, then $x_{\ell+1} \in \X_2(\|\{e_i\}\|)$.
By induction, we have thus shown that if $x_\ell \in \X_2(\|\{e_i\}\|)$ for some $\ell\in\N_0$, then $x_k \in \X_2(\|\{e_i\}\|)$ for all $k\ge \ell$.\mer


Let $|x_0|\leq M_0$ and $t_k=\sum_{i=0}^{k-1} T_i$ for every $k\in \N_0$. 
Consider the function
\begin{multline}
  \label{eq:yTdefn21} 
  y(t):=V_{M,E}(x_{k})+ \frac{t-t_k}{T_{k}}   \left[ V_{M,E}(x_{k+1})-V_{M,E}(x_{k}  ) \right] \\
                  \text{if } t \in \left[t_k,t_{k+1}\right),
\end{multline}
which depends on the initial condition $x_0$, on the sampling period sequence
$\{T_i\}$, on the disturbance sequence $\{e_i\}$ and on the given constants $M,E$
(through the fact that $V$ depends on the latter constants) and satisfies $y(0)=V_{M,E}(x_0)\geq0$.
From \eqref{eq:yTdefn21} we have that 
\begin{align}
\label{eq:doty} 
\dot{y}(t)= \frac{V(x_{k+1})-V(x_k)}{T_k}  \quad \forall t \in \left(t_k,t_{k+1}\right), \forall k \in \N_0 
\end{align}
and
\begin{align}
 y(t) &\leq V_{M,E}(x_{k}),\quad   \forall t \in \left[t_k,t_{k+1}\right).
  \label{eq:yProperty11}
  \end{align}
By  Claim~\ref{clm:3} and \eqref{eq:X2}, we have that \eqref{eq:diffVkp2}
holds for all $x_k \notin \X_2(\|\{e_i\}\|)$ for all $k\in \N_0$.
Using \eqref{eq:diffVkp2} and \eqref{eq:doty},
for all
$x_k \notin \X_2(\|\{e_i\}\|)$, we have 
\begin{align}
\label{eq:yProperty2}
   \dot{y}(t) &\leq -\alpha(V_{M,E}(x_{k})) \leq -\alpha(y(t) ).
\end{align}
Hence (\ref{eq:yProperty2}) holds for almost all $t\in [0,t_{k^*})$ with $t_{k^*}=\inf \{t_k: x_k \in \X_2(\|\{e_i\}\|) \}$.
Note that the function $\alpha = \alpha_3 \comp \alpha_2^{-1}$ does not depend on any of the following quantities: $x_0$, $\{T_i\}$, $\{e_i\}$, $M_0$ or $E_0$.
Using Lemma~4.4 of \cite{LinSontagWangSIAM96}, there exists $\beta_1 \in \KL$ such that,
for all $t\in[0,t_{k^*})$ we have
\begin{align}
\label{eq:ult}
y(t) \leq \beta_1\left(y(0),t\right). 
\end{align}
From \eqref{eq:yTdefn21}, $y(t_k)=V_{M,E}(x_k)$ for all $k\in \N_0$.
Evaluating \eqref{eq:ult} at $t=t_k$, then
\begin{align}
\label{eq:54}
 V_{M,E}(x_k) \leq \beta_1\left(V_{M,E}(x_0),\sum_{i=0}^{k-1} T_i\right), \quad k=0,1,\hdots,k^*-1.
\end{align}
From Claim~\ref{clm:4} and \eqref{eq:X2} if $x_{k} \in  \X_2(\|\{e_i\}\|)$ 
then $V_{M,E}(x_k)\leq \alpha_2 \comp \eta(\|\{e_i\}\|)$.
Combining the latter with \eqref{eq:54}, then 
\begin{align*}
 V_{M,E}(&x_k) \leq \beta_1\left(V_{M,E}(x_0),\sum_{i=0}^{k-1} T_i\right)+  \alpha_2\left( \eta(\|\{e_i\}\|)\right), \quad \forall k\in  \N_0.
\end{align*}
Define $\beta \in \KL$ via
$\beta(s,\tau):=\alpha_1^{-1}(2\beta_1(\alpha_2(s),\tau))$, 
and $\gamma \in \K_\infty$ via
$\gamma(s):= \alpha^{-1}_1(2\alpha_2(\eta(s)))$.
Using the fact that $\chi(a+b) \le \chi(2a)+\chi(2b)$ for every $\chi\in\K$
and \eqref{eq:cond_1_semiglobal}
it follows that
\begin{align}
|x_k| &\leq 
   \beta\left(|x_0|,\sum_{i=0}^{k-1} T_i\right)+\gamma(\|\{e_i\}\|)
\end{align}
for all $|x_0|\leq M_0$, all $\|e_i \| \leq E_0$
and $\{T_i\} \in \Phi (T^\star)$.
We have thus established that \eqref{eq:system1} is S-ISS-VSR.
\end{proof}

\begin{color}{black}
Theorem~\ref{lem:ISS-VSR=LF} shows that there is no loss of generality in 
the search of a Lyapunov function for a S-ISS-VSR discrete-time model 
since its existence is a necessary condition. The fact that S-ISS-VSR implies
SP-ISS-VSR then shows that Theorem~\ref{lem:ISS-VSR=LF} also provides  sufficient, 
althought not necessary, conditions for SP-ISS-VSR.
In \cite[Theorem~3.2]{valhai_arxiv18} we provided 
checkable sufficient conditions for a discrete-time
model of the form \eqref{eq:fut}, \eqref{eq:system1} to be SP-ISS-VSR.
The conditions in items i), ii) and iii) for \cite[Theorem~3.2]{valhai_arxiv18}
and Theorem~\ref{lem:ISS-VSR=LF} above are identical.
The main difference between these theorems reside in the Lyapunov-type condition:
the quantity $R>0$ that defines the practical nature of
the SP-ISS-VSR property does not exist here and 
the Lyapunov function of the current theorem is only upper bounded
in a compact set defined by $M \geq 0$.
The existence of 
necessary and sufficient conditions
of the kind of Theorem~\ref{lem:ISS-VSR=LF} 
for the SP-ISS-VSR property remains an open problem.
\end{color}

\section{Example}
\label{sec:examples}

Consider the Euler (approximate) discrete-time model 
of the Example A of \cite{valhai_arxiv18}:
\begin{align}
x_{k+1}=x_k+T_k(x^3_k+u_k)=:F(x_k,u_k,T_k). 
\end{align}
This open-loop Euler model was fed back with the control law
$u_k=U(\hat x_k,T_k)=-\hat x_k-3\hat x_k^3$ and 
additive state-measurement errors $e_k$ were considered,
so that $\hat x_k=x_k+e_k$.
The resulting approximate closed-loop model
$\bar F(x,e,T)=F(x,U(x+e,T),T)$ is
\begin{align}
\label{eq:appmodelerror}
\bar F(x,e,T)=x-T[2x^3+9ex^2+(9e^2+1)x+3e^3+e].
\end{align}
In Example A of \cite{valhai_arxiv18} we established
that \eqref{eq:appmodelerror} is SP-ISS-VSR with respect
to input $e$. We will prove that
\eqref{eq:appmodelerror} is not only SP-ISS-VSR
but also S-ISS-VSR. We make use of 
Theorem~\ref{lem:ISS-VSR=LF}).
The continuity and boundedness assumptions
\ref{enum:lemequiv1}), \ref{enum:lemequiv3}) and \ref{enum:lemequiv2})
of Theorem~\ref{lem:ISS-VSR=LF})
are easy to verify for \eqref{eq:appmodelerror}. 
To prove assumption \ref{item:lyapunov}) define 
$\alpha_1,\alpha_2,\alpha_3, \rho \in \K_\infty$
via $\alpha_1(s)=\alpha_2(s)=s^2$,
$\alpha_3(s)=3s^4 +s^2$ and $\rho(s)=s/K$ with $K>0$ to be selected.
Let $M\ge0$ and $E\ge0$ be given 
and define $V(x)=x^2$.
Then 
\eqref{eq:cond_1_semiglobal}
is satisfied.
 We have
\begin{align}
  \label{eq:Lyap_1}
  &V(\bar F(x,e,T))-V(x) = \left[h(x,e)+g(x,e)T\right]T, \\
  &h(x,e) = -2x[2 x^3 + 9ex^2 + (9e^2 + 1)x + (3e^3 + e)],\\
  &g(x,e) = [2 x^3 + 9ex^2 + (9e^2 + 1)x + (3e^3 + e)]^2.
\end{align}
Expanding $g(x,e)$, taking absolute values on sign indefinite
terms and noting that whenever $\rho(|e|)\leq |x|$ we have 
$|e|\leq K|x|$ we can bound $g(x,e)$ as
\begin{align*}
g(x,e)\leq a(K)x^6+b(K) x^4+c(K) x^2, \quad \text{if } |e|\leq K|x|,   
\end{align*}
where $a(K)=9 K^6+135 K^4 +174 K^3+117 K^2+90 K +4$,
$b(K)= 6 K^4+24 K^3+36 K^2+22K+4$
and $c(K)=K^2+2K+1$.
Expanding $h(x,e)$, taking absolute values
on sign indefinite terms and bounding
$|e|\leq K|x|$, it follows that 
\begin{align*}
  h(x,e) &\leq -4x^4 - 2 x^2 + 18 K x^4 + (6K^3 |x|^3 + 2K |x|)|x| \\
  &\leq -4x^4 -2x^2 +d(K)x^4 + 2K x^2 
\end{align*}
where $d(K)=(6K^2+18)K$.
Select $K=0.025$ and
$\tilde T= \min \left\{\frac{1}{2b(K)}, \frac{1}{2(a(K)M^4+c(K))} \right\}$.
Then, for all $\frac{|e|}{K}\leq |x|\leq M$,
we can bound \eqref{eq:Lyap_1} as
\begin{align*}
&h(x,e)+g(x,e)T \\
&\leq -4x^4 -2x^2 +d(K)x^4 + 2K x^2 \\
& \qquad\qquad\quad\medspace\medspace\medspace+(a(K)x^6+b(K) x^4+c(K) x^2)T \\
&= -3x^4 -x^2 \\
&\quad +x^2\left((b(K)T+d(K)-1)x^2+a(K)x^4T+c(K)T+2K-1\right) \\
&\leq -\alpha_3(|x|) \\
& \quad +x^2\left((b(K) \tilde T+d(K)-1)x^2+(a(K)M^4+c(K))  \tilde T+2K-1\right)  \\
&\leq -\alpha_3(|x|).
\end{align*}
The last inequality holds
because, for the chosen values of $K$ and $\tilde T$,
the expression between parentheses is
less than zero.
Thus, assumption \ref{item:lyapunov}) of Theorem~\ref{lem:ISS-VSR=LF}
is satisfied and the system \eqref{eq:appmodelerror} is S-ISS-VSR.

\section{Conclusions}
\label{sec:conclusions}

We have given necessary and sufficient conditions for two stability properties especially suited to discrete-time models of nonlinear systems under non-uniform sampling. We have given $\epsilon$-$\delta$ and Lyapunov-based characterizations
of robust semiglobal stability (RSS-VSR) and a Lyapunov-type characterization of semiglobal input-to-state stability (S-ISS-VSR), both under \begin{color}{black} non-uniform sampling\end{color}.
We have illustrated the application of the results
on a numerical example for an approximate closed-loop discrete-time model
with additive state measurement disturbances. The \begin{color}{black}provided results\end{color} can be combined
with previous results to ensure stability properties for closed-loop systems whose control law has been designed based on an approximate model.

\appendix

\begin{color}{black}
\section{Discrete-time Model Existence Conditions}
\label{app:dt-mod}

The exact discrete-time model for a given continuous-time nonlinear system \eqref{eq:cs} is the discrete-time system whose state matches the state of the continuous-time system at every sampling instant. If a discrete-time system of the form (\ref{eq:fut}) is the exact model of the system~(\ref{eq:cs}) under non-uniform sampling and zero-order hold, then this fact implies that the function $f$ is such that (\ref{eq:cs}) admits a unique solution from every initial condition $x_0 \in \R^n$. An exact discrete-time model of the form (\ref{eq:fut}) need not be defined for all $(x_k,u_k,T_k) \in \R^n \times \R^m \times \R_{>0}$, even if $f$ in (\ref{eq:cs}) satisfies $f : \R^n \times \R^m \rightarrow \R^n$, i.e. even if $f$ is globally defined, because the solution to (\ref{eq:cs}) with constant $u$ may be not defined for all $t\ge 0$ (the solution may have finite escape time). The following lemma shows that under reasonable boundedness and Lipschitz continuity conditions on $f$, the exact discrete-time model will exist.
\begin{aplemma}
  \label{assu1}
  Let $f : \R^n \times \R^m \to \R^n$ satisfy 
  \begin{enumerate}[a)]
  \item For every pair of compact sets $\X \subset \R^n$ and $\U \subset \R^m$, there exists $C>0$ such that $|f(x,u)|\leq C$ for all $x \in \X$ and $u\in \U$.\label{item:fbnd}
  \item For every compact set $\mathcal{X}\subset \R^n$ and $u\in \R^m$ there exists $L:=L(\X,u)>0$ such that for all $x,y \in \X$,\label{item:Lip}
    \begin{align*}
      |f(x,u)-f(y,u)| &\leq L|x-y|.
    \end{align*}
  \end{enumerate}
  Then, for every $x_0 \in \R^n$ and constant $u(t) \equiv u\in\R^m$, (\ref{eq:cs}) admits a unique maximal (forward) solution $\phi_u(t,x_0)$, defined for all $0\le t< T(x_0,u)$ with $0 < T(x_0,u) \le \infty$. Moreover, for every pair of compact sets $\tilde{\X} \subset \R^n$ and $\U \subset \R^m$, there exists $T^*>0$ such that $T(x_0,u) \ge T^*$ for every $(x_0,u)\in \tilde{\X}\times\U$.
\end{aplemma}
\begin{proof}
  Existence and uniqueness of the solution follows from standard results on differential equations (e.g. \cite{hale_book80}) given the Lipschitz continuity assumption \ref{item:Lip}) and noticing that the solution corresponds to a constant $u(t)$ in (\ref{eq:cs}).

  Next, consider compact sets $\tilde{\mathcal{X}}\subset\R^n$ and $\mathcal{U}\subset\R^m$. In correspondence with $\tilde{\X}$ define $r := \max\big\{1,\max_{x\in\tilde{\X}} |x|\big\}$ and $\X := \{x\in\R^n : |x| \le 2r\}$. Since $\tilde\X$ is compact, then it is bounded and $r<\infty$. Let condition~\ref{item:fbnd}) generate $C>0$ in correspondence with $\mathcal{X}$ and $\mathcal{U}$. Define $T^*:=r/C$ and note that $T^* > 0$ because $r>0$ and $C>0$. Let $x_0\in \tilde{\X}$, $u\in \mathcal{U}$
and let $\phi_u(t,x_0)$ denote the unique solution to (\ref{eq:cs}), where $u(t) \equiv u$ is constant. Then
\begin{align*}
\phi_u(t,x_0)&= x_0+ \int_{0}^t {f(\phi_u(s,x_0),u)} ds 
\end{align*}
Considering condition~\ref{item:fbnd}), we have, for all $t\in [0,T^*)$,
\begin{align*}
|\phi_u(t,x_0)|& \leq |x_0| + \int_{0}^t  |{f(\phi_u(s,x_0),u)} | ds \leq r + C t < 2r
\end{align*}
and hence $\phi_u(t,x_0)$ exists and remains within $\X$ for all $t\in [0,T^*)$. This establishes that the maximal (forward) existence time $T(x_0,u)$ for the solution $\phi_u(t,x_0)$ is not less than $T^*$. 
\end{proof}

\end{color}
\section{Proof of Lemma~\ref{lemma:equiva}}
\label{A}

\textbf{($\ref{it:RSSVSR}.\Rightarrow2.$)} 
Let $\beta\in\KL$ be given by the RSS-VSR property. Consider $M> 0$ and let $T^\blacktriangle = T^\blacktriangle(M)>0$. 

Let $\alpha\in\Ki$ be defined via $\alpha(s) := \beta(s,0)$. Let $\epsilon > 0$ and take $\delta = \alpha^{-1}(\epsilon) > 0$. 
Let $x_k$ denote a solution to (\ref{eq:sysd}) satisfying $|x_0| \le \min \{ \delta, M\}$ and corresponding to 
$\{T_i\} \in \Phi(T^\blacktriangle)$ and $\{d_i\} \in \D$. From~(\ref{eq:semiglobal_cond}), we have $|x_k|\leq \beta(|x_0|,0) \le \beta(\delta,0) = \epsilon$, for all $k \in \N_0$. Then, \ref{enum:lemequiv1}) holds.  

Define $C(M,L) := \beta(M,0)$. Let $x_k$ denote a solution to (\ref{eq:sysd}) satisfying $|x_0| \leq M$
and corresponding to $\{T_i\} \in \Phi(T^\blacktriangle)$  and $\{d_i\} \in \D$. From~(\ref{eq:semiglobal_cond}), then $|x_k| \le \beta(|x_0|,0) \le \beta(M,0) = C(M,L)$ for all $k\in \N_0$. 
Then, \ref{enum:lemequiv3}) holds.

Let $\epsilon>0$ and select $\mathcal{T}\ge 0$ such that $\beta(M,\mathcal{T}) \le \epsilon$. Let $x_k$ denote a solution to (\ref{eq:sysd}) 
satisfying $|x_0|\le M$, $\{T_i\} \in \Phi(T^\blacktriangle)$ and $\{d_i\} \in \D$. From~(\ref{eq:semiglobal_cond}), we have $|x_k|\leq \beta(|x_0|,\sum_{i=0}^{k-1}T_i) \le \beta(M,\mathcal{T}) \le \epsilon$,
for all $k\in \N_0$ for which $\sum_{i=0}^{k-1}T_i\geq \mathcal{T}$. Then, \ref{enum:lemequiv2}) holds.

\textbf{($1.\Leftarrow 2.$)} 
Let $M\geq0$ and $T^\blacktriangle(M) > 0$ be such that conditions \ref{enum:lemequiv1})--\ref{enum:lemequiv2}) hold. 
Let $\bar\delta(\epsilon) := \sup \{\delta : \delta \text{ corresponds to $\epsilon$ as in \ref{enum:lemequiv1})}\}$, the supremum of all applicable $\delta$. Then $\bar\delta(\epsilon) \le \epsilon$ for all $\epsilon > 0$, 
and $\bar \delta : \R_{> 0} \to \R_{>0}$ is positive and non-decreasing. Let $0 < c < 1$. Then, there exists $\alpha\in \K$ such that 
$\alpha(s)\leq c \bar \delta(s)$. 
Define $c_1 =\lim_{s \rightarrow \infty} \alpha(s)$, then $\alpha^{-1}:[0,c_1)\rightarrow \R_{\ge0}$. 
From \ref{enum:lemequiv1}), we know that
$|x_0|\leq \alpha(\epsilon) \le \min \{c \bar \delta(\epsilon),M\} < \bar\delta(\epsilon) \medspace  \Rightarrow \medspace |x_k|\leq \epsilon$, 
for all  $k \in \N_0$, all $\{T_i\} \in \Phi(T^\blacktriangle(M))$ and all $\{d_i\} \in \D$.
Choosing $\epsilon=\alpha^{-1}(|x_0|)$ when $|x_0| < c_1$, 
it follows that whenever $|x_0| < c_1$ and $|x_0|\leq M$, $k\in\N_0$, $\{T_i\} \in \Phi(T^\blacktriangle(M))$ and $\{d_i\} \in \D$, then
\begin{align}
  \label{eq:alpinv1}
  |x_k| &\le \alpha^{-1}(|x_0|).
\end{align}
Next, define 
\begin{align*}
  \underbar{$C$}(M,L) &:= \inf \{ C : C \text{ corresponds to }M,L \text{ as in \ref{enum:lemequiv3})} \},\\
  \underbar{$\mathcal{T}$}(M,\epsilon) &:= \inf \{\mathcal{T} : \mathcal{T} \text{ corresponds to }M,\epsilon \text{ as in \ref{enum:lemequiv2})} \},
\end{align*}
the infima over all applicable $C$ and $\mathcal{T}$ from conditions \ref{enum:lemequiv3}) and \ref{enum:lemequiv2}), respectively. 
Then, $ \underbar{$C$}$ is nonnegative and nondecreasing in each variable, 
and $  \underbar{$\mathcal{T}$}(M,\epsilon)$ is nonnegative, nondecreasing in $M$ for every fixed $\epsilon > 0$,
and nonincreasing in $\epsilon$ for every fixed $M>0$. 
Given $s>0$, consider $  \underbar{$\mathcal{T}$}(s,1/s)$ 
and $\underbar{$C$}(s,\underbar{$\mathcal{T}$}(s,1/s))$. 
If $x_k$ is a solution to (\ref{eq:sysd}) 
corresponding to an initial condition satisfying $|x_0| \le s$,
$\{T_i\} \in \Phi(T^\blacktriangle(s))$ and $\{d_i\} \in \D$, then
\begin{align}
  \label{eq:xkbndall}
  |x_k| &\le
          \begin{cases}
             \underbar{$C$}(s,\underbar{$\mathcal{T}$}(s,1/s)), &\text{ whenever }\sum_{i=0}^{k-1} T_i < \underbar{$\mathcal{T}$}(s,1/s),\\
            \frac{1}{s}, &\text{ otherwise.}
          \end{cases}
\end{align}
Define $\hat C : \R_{>0} \to \R_{>0}$ via $\hat C(s) := \max\{ \underbar{$C$}(s,\underbar{$\mathcal{T}$}(s,1/s)), 1/s\}$. By the monotonicity properties of \underbar{$C$}, \underbar{$\mathcal{T}$} and $1/s$, there exists $p > 0$ such that $\hat C$ decreases over $(0,p)$ and is nondecreasing over $(p,\infty)$. Therefore there exists $\bar\alpha\in\Ki$ such that 
\begin{align}
  \label{eq:alfabarra}
  \bar \alpha (s) &\ge
          \begin{cases}
             \alpha^{-1}(s), \medspace &\text{if  }   0 \leq s < \frac{c_1}{2},\\
            \hat C(s), \medspace &\text{if  }   \frac{c_1}{2} \leq s.  
          \end{cases}
\end{align}
Then, if $\{T_i\} \in \Phi(T^\blacktriangle(|x_0|))$ and $\{d_i\} \in \D$, by \eqref{eq:alpinv1}--(\ref{eq:alfabarra}), we have $|x_k| \le \bar\alpha(|x_0|)$ for all  $k\in \N_0$.
Consequently, if $\{T_i\} \in \Phi(T^\blacktriangle(M))$ and $\{d_i\} \in \D$ 
\begin{align}
  \label{eq:cond1_equiv}
  \bar\alpha(M) \le \epsilon \text{ and } |x_0| \le M \medspace \medspace \Rightarrow \medspace\medspace |x_k| \le \epsilon, \medspace \forall k\in \N_0,
\end{align}
and, by \ref{enum:lemequiv2}),
\begin{align}
  \label{eq:x0undT}
  \sum_{i=0}^{k-1}T_i >  \underbar{$\mathcal{T}$}(M,\epsilon) \text{ and } |x_0|\le M \medspace \medspace \Rightarrow \medspace \medspace
  |x_k|\leq \epsilon.
\end{align}
Define $\tilde\beta : \R_{\geq 0} \times \R_{\geq 0} \to \R_{\ge 0}$ via
\begin{align}
  \label{eq:defbeta1}
  \tilde\beta(r,t) &:= \inf \{\epsilon : \underbar{$\mathcal{T}$}(r,\epsilon) \leq t \}.
\end{align}

\begin{claim}
  \label{clm:beta}
  There exists $\beta\in\KL$ such that $\beta \geq \tilde \beta$. 
\end{claim}
\indent\emph{Proof of Claim~\ref{clm:beta}:}
We will prove that  $\tilde \beta$ satisfies the 
conditions of \cite[Lemma~15]{Albertini99}.

Consider $ \epsilon_1>0$ given. Let $0<c<1$ and choose
$\delta:=  \bar \alpha^{-1}(c\epsilon_1)$.
Then, from \eqref{eq:cond1_equiv},
for all $r:=|x_0| \leq \delta$ we have that 
$|x_k|\leq c \epsilon_1$ for all $k\in \N_0$, all $\{T_i\} \in \Phi(T^\blacktriangle(\delta))$ and all $\{d_i\} \in \D$,
thus $ \underbar{$\mathcal{T}$}(r, \epsilon_1)=0$.
Therefore $\tilde \beta(r,t) \leq c \epsilon_1 < \epsilon_1$ for all $0 \leq  r \leq \delta$
and all $t\geq 0$.

Consider $ \epsilon_2>0$ and  $M_2>0$ given.
Let $0<c<1$, then $\tilde{\mathcal{T}}:= \underbar{$\mathcal{T}$}( M_2,c \epsilon_2)\geq\underbar{$\mathcal{T}$}(r,\epsilon)$
for all $0 \leq r\leq  M_2$ and $\epsilon \geq c \epsilon_2$.
By \eqref{eq:x0undT}, for all $\{T_i\} \in \Phi(T^\blacktriangle( M_2))$ and all $\{d_i\} \in \D$, we have
$\tilde\beta(r, t)=\inf \{\epsilon : \underbar{$\mathcal{T}$}(r,\epsilon) \leq t \} \leq c \epsilon_2< \epsilon_2$
for all $0 \leq r\leq  M_2$ and all $t\geq \tilde{\mathcal{T}}$.

By \cite[Lemma~15]{Albertini99} there exists $\beta \in \KL$ such that
$\tilde \beta(r,t) \leq \beta (r,t)$ for all $r\geq0$ and all $t\geq0$ . \mer

From~\eqref{eq:x0undT} and (\ref{eq:defbeta1}) 
then if $|x_0| \le M$ we have, for all $k \in \N_0$ such that
$\sum_{i=0}^{k-1} T_i \ge t$
with $\{T_i\} \in \Phi(T^\blacktriangle( M))$ and $\{d_i\} \in \D$,   
that $|x_k| \le \tilde\beta(M,t) \le \beta(M,t)$. 
Consequently,
\begin{align*}
  |x_k| &\le \beta\left(|x_0|,\sum_{i=0}^{k-1} T_i \right),
\end{align*}
for all $k\in \N_0$, all $\{ T_i\} \in \Phi(T^\blacktriangle(M))$, all $|x_0|\leq M$
and all $\{d_i\} \in \D$,
which establishes that the system \eqref{eq:sysd} is RSS-VSR. 


(\ref{it:RSSVSR}. $\Rightarrow$ \ref{it:boundedV}.) 
Let $\beta \in \KL$ and $T^\blacktriangle(\cdot)$ be given by the RSS-VSR property. Without loss of generality, suppose that $T^\blacktriangle$ is nonincreasing (recall Remark~\ref{rem:Ttriangle}).
It follows from \cite[Lemma~7]{Kellett2014} that there exist $\alpha_1, \alpha_2 \in \K_\infty$ such that
\begin{align}
  \label{eq:cond_1_sglobal}
  \alpha_1(\beta(s,t)) \leq \alpha_2(s) e^{-3t}, \quad \forall s\geq 0, t\geq0.
\end{align}
Define $\alpha_3 := \alpha_1$, let $\bar{T}>0$ be such that
\begin{align}
  \label{eq:condtime2}
  T\leq 1-e^{-2T}, \quad \forall T\in (0,\bar{T}),
\end{align}
and define $T^* : \R_{>0} \to \R_{>0}$ via $T^*(s) :=\min \{ T^\blacktriangle(\beta(s,0)), \bar{T} \}$. Since $s \le \beta(s,0)$ for all $s\ge 0$ and $T^\blacktriangle$ is nonincreasing, it follows that $T^*(s) \le T^\blacktriangle(s)$, 
and hence $\Phi(T^*(s)) \subset \Phi(T^\blacktriangle(s))$, for all $s > 0$. Let $x(k,\xi,\{d_i\},\{T_i\})$ denote the solution of \eqref{eq:sysd} 
at instant $k \in \N_0$ that corresponds to a sampling period sequence $\{T_i\}$, the disturbance sequence $\{d_i\}$ and 
satisfies $x(0,\xi,\{d_i\},\{T_i\}) = \xi$. 
Note that $x(k,\xi,\{d_i\},\{T_i\})$ may be not defined for arbitrary $(k,\xi,\{d_i\},\{T_i\})$.
For $M\ge0$ and $\xi\in\R^n$, define the set
\begin{multline}
  \label{eq:exS}
  S(M,\xi) := \big\{ (k,\{d_i\},\{T_i\}) \in \N_0 \times \mathcal{D} \times\Phi(T^*(M)):   \\  x(k,\xi,\{d_i\},\{T_i\}),  \text{ is defined}\big\},
\end{multline}
and the function $V_M : \R^n \to \R_{\ge 0} \cup \{\infty\}$ via
\begin{align}
  \label{eq:LF}
  V_M(\xi):=\sup_{(k,\{d_i\},\{T_i\}) \in S(M,\xi)} \alpha_1( |x(k,\xi,\{d_i\},\{T_i\})|) e^{2\sum_{i=0}^{k-1}T_i}.
\end{align}
Note that if $\{d_i\} \in \mathcal{D}$ and $\{T_i\} \in \Phi(T^*(M))$, then $(0,\{d_i\},\{T_i\}) \in S(M,\xi)$ for all $\xi\in\R^n$ because $x(0,\xi,\{d_i\},\{T_i\}) = \xi$ is defined. 
Hence, $S(M,\xi) \neq \emptyset$ for all $M\ge0$ and $\xi\in\R^n$. We next show that item~\ref{it:boundedV}) is satisfied with $\alpha_1,\alpha_2,\alpha_3,T^*$ 
and $V_M$ as defined.

From (\ref{eq:LF}), it follows that
\begin{align}
  \label{eq:LF1}
  V_M(\xi) \geq \alpha_1(|x(0,\xi,\{d_i\},\{T_i\})|)=\alpha_1(|\xi|),
\end{align}
holds for all $\xi \in \R^n$. 
Then \eqref{eq:cond_1_semiglobal_a32} follows.
Consider $M\ge 0$. 
For $|\xi| \le M$ and $(k,\{d_i\},\{T_i\}) \in S(M,\xi)$, 
it follows that $\{T_i\} \in \Phi(T^*(M)) \subset \Phi(T^\blacktriangle(M))$. 
Therefore, $|\xi| \le M$ and $(k,\{d_i\},\{T_i\}) \in S(M,\xi)$ imply that (\ref{eq:semiglobal_cond}) is satisfied. 
Using (\ref{eq:semiglobal_cond}) and~\eqref{eq:cond_1_sglobal} in \eqref{eq:LF} then, for all $|\xi|\leq M$ we have
\begin{align}
  \label{eq:LF2}
  V_M(\xi) &\leq \sup_{(k,\{d_i\},\{T_i\}) \in S(M,\xi)} \alpha_1\left[ \beta\left(|\xi|,\sum_{i=0}^{k-1}T_i\right)\right] e^{2\sum_{i=0}^{k-1}T_i} \notag \\
           &\leq \sup_{(k,\{d_i\},\{T_i\}) \in S(M,\xi)} \alpha_2(|\xi|)e^{-\sum_{i=0}^{k-1}T_i} = \alpha_2(|\xi|), 
\end{align}
whence \eqref{eq:cond_1_semiglobal_b32} is established. 
Let $M\ge 0$ and define $\tilde M := \beta(M,0)$. By the definition of RSS-VSR and the fact that $\Phi(T^*(M)) \subset \Phi(T^\blacktriangle(\tilde M)) \subset \Phi(T^\blacktriangle(M))$, it follows that 
$|x(j,\xi,\{d_i\},\{\bar T_i\})| \leq \tilde M$ for all $j \in \N_0$ whenever $|\xi| \le M$, $\{d_i\} \in \mathcal{D}$ and $\{\bar T_i\} \in \Phi(T^*(M))$.
Therefore, 
$S(M,x(j,\xi,\{d_i\},\{\bar T_i\})) = \N_0 \times \mathcal{D} \times\Phi(T^*(M))$ for all $j \in \N_0$ whenever $|\xi|\le M$, $\{d_i\} \in \mathcal{D}$
and $\{\bar T_i\} \in \Phi(T^*(M))$. 
Thus, for all $|\xi| \leq M$, all $|d|\leq D$ and all $T\in (0,T^*(M))$, we have $S(M,\bar F(\xi,d,T)) = \N_0 \times \mathcal{D}  \times \Phi(T^*(M))$, and
\begin{align}
  \lefteqn{V_{M}(\bar F(\xi,d,T)) =}\hspace{3mm} \notag \\
&=\sup_{(k,\{d_i\},\{T_i\}) \in \N_0 \times \mathcal{D}\times\Phi(T^*(M))} \alpha_1\left(\big|{\scriptstyle x\big(k,\bar F(\xi,d,T),\{d_i\},\{T_i\}\big)}\big|\right)e^{2\sum_{i=0}^{k-1}T_{i}} \notag\\
&=\sup_{(\ell,\{d_i\},\{T_i\}) \in \N \times \mathcal{D}\times \Phi(T^*(M))} \alpha_1\left( \big|{\scriptstyle x \big(\ell, \xi,\big\{d,\{d_i\} \big\},\big\{T,\{T_i\} \big\} \big)} \big| \right) e^{2 \sum_{i=0}^{\ell-2}T_i} \notag \\
&\leq e^{-2T}\sup_{(\ell,\{d_i\},\{T_i\}) \in \N_0 \times \mathcal{D} \times \Phi(T^*(M))} \alpha_1(|x(\ell,\xi,\{d_i\},\{T_i\})|)e^{2\sum_{i=0}^{\ell-1}T_i} \notag \\
  \label{eq:Vdec}
&\le V_{M}(\xi) e^{-2T}  
\end{align}
where we have used the facts that $e^{2\sum_{i=0}^{l-1}T_i} = e^{2(T+\sum_{i=0}^{l-1}T_i)} e^{-2T}$,
 $\big\{ d, \{d_i\} \big\} \in \mathcal{D}$ and $\big\{ T, \{T_i\} \big\} \in \Phi(T^*(M))$. 
From  \eqref{eq:condtime2}, \eqref{eq:Vdec},
and the fact that $T\in (0,T^*(M)) \subset (0,\bar T)$, then for $|\xi| \le M$ we have
\begin{align*}
  &V_M(\bar F(\xi,d,T)) \leq V_M(\xi)-V_M(\xi)(1-e^{-2T}),\quad\text{hence} \\
  &V_M(\bar F(\xi,d,T))-V_M(\xi) \leq -TV_M(\xi) \le -T \alpha_3(|\xi|),
\end{align*}
where we have used \eqref{eq:LF1}. Then \eqref{eq:cond_2_semiglobal_32} follows and
(\ref{it:RSSVSR}. $\Rightarrow$ \ref{it:boundedV}.) is established.

(\ref{it:boundedV}. $\Rightarrow$ \ref{it:RSSVSR}.)
Define $T^\blacktriangle = T^* \comp \alpha_1^{-1}\comp\alpha_2$ and $\alpha:=\alpha_3 \comp \alpha_2^{-1}$. Let $\beta_1 \in \KL$ correspond to $\alpha$ as per Lemma~4.4 of \cite{LinSontagWangSIAM96}, and define $\beta\in\KL$ via
\begin{align}
  \label{eq:defbeta}
  \beta(s,t) &= \alpha_1^{-1}\big( \beta_1(\alpha_2(s),t) \big).
\end{align}
We next show that $\beta$ and $T^\blacktriangle$ as defined characterize the RSS-VSR property of (\ref{eq:sysd}). Let $\check M \ge 0$ and consider $\{T_i\} \in \Phi(T^\blacktriangle(\check M))$ and $|x_0| \le \check M$. We have to show that (\ref{eq:semiglobal_cond}) holds for all $k\in\N_0$ when $x_k := x(k,x_0,\{d_i\},\{T_i\})$ denotes the solution to \eqref{eq:sysd} corresponding to the initial condition $x_0$, disturbance sequence $\{d_i\}$ and the sampling period sequence $\{T_i\}$. 
Define $M := \alpha_1^{-1} \comp \alpha_2(\check M)$, so that $T^\blacktriangle(\check M) = T^*(M)$. Note that $M\ge \check M$. Define
\begin{align}
  \X_M &:=\{ x\in\R^n : V_M(x) \leq \alpha_1(M)\}.
\end{align}
From \eqref{eq:cond_1_semiglobal_a32}, 
if $x_k$ satisfies $V_M(x_k) \le \alpha_1(M)$, then $|x_k| \le M$, 
and from \eqref{eq:cond_2_semiglobal_32}, then $V_M(x_{k+1}) \le V_M(x_k) \le \alpha_1(M)$. 
Since $x_0$ satisfies $|x_0| \le \check M$,
from \eqref{eq:cond_1_semiglobal_b32} then $V_M(x_0)\leq \alpha_2(\check M)=\alpha_1(M)$, and hence
$x_k \in \X_M$ and $|x_k| \le M$ hold for all $k\ge 0$.
From~\eqref{eq:cond_1_semiglobal_b32}, it follows that $\alpha_2^{-1}(V(x_k))\leq |x_k|$. Using this inequality in \eqref{eq:cond_2_semiglobal_32}, then
\begin{align}
  \label{eq:diffVkp1}
  V_M(x_{k+1}) - V_M(x_{k}) &\le -T_k \alpha(V_M(x_{k})). 
\end{align}
Let $t_k = \sum_{i=0}^{k-1} T_i$ for every $k\in\N_0$. Consider the function
\begin{multline}
  \label{eq:yTdefn} 
  y(t) := V_M(x_{k}) + \frac{t- t_k}{T_{k}} \left[V_M(x_{k+1}) - V_M(x_{k}) \right],\\
  \text{if } t \in \left[t_k,t_{k+1}\right).
\end{multline}
Note that the function $y(\cdot)$ depends on the initial condition $x_0$,
on the disturbance sequence $\{d_i\}$,
on the sampling period sequence $\{T_i\}$ and on $M$ (through $V_M$), and satisfies $y(0) = V_M(x_0) \ge 0$. From~(\ref{eq:yTdefn}), it follows that
\begin{align}
  \dot{y}(t) = \frac{V_M(x_{k+1}) - V_M(x_{k})}{T_k},  \quad \forall t \in (t_k,t_{k+1}), \forall k\ge 0.
\end{align}
Using \eqref{eq:diffVkp1} and  \eqref{eq:yTdefn} we have , for $|x_0| \leq \check M$, 
\begin{align}
  y(t) &\leq V_M(x_{k}),\quad   \forall t \in \left[t_k,t_{k+1}\right),\quad \text{and}\\
  \label{eq:yProperty}
  \dot{y}(t) &\leq -\alpha(V_M(x_{k})) \leq -\alpha(y(t) ),
\end{align}
where (\ref{eq:yProperty}) holds
for almost all $t\in [0,\T_{\{T_i\}})$, with $\T_{\{T_i\}} := \sum_{i=0}^\infty T_i \in \R_{>0} \cup \{\infty\}$.
By Lemma~4.4 of \cite{LinSontagWangSIAM96}, then for all $t\in[0,\T_{\{T_i\}})$, we have
\begin{align}
  \label{eq:yineq}
  y(t) &\leq \beta_1(y(0),t).
\end{align}
From (\ref{eq:yTdefn}), $y(t_k) =V_M(x_{k})$ for all $k\in\N_0$. Evaluating (\ref{eq:yineq}) at $t=t_k$, then
\begin{align}
  \label{eq:cond_v1}
 V_M(x_{k})\leq \beta_1\left( V_M(x_0),\sum_{i=0}^{k-1} T_i \right),\quad \forall k\in\N_0.
\end{align}
Using 
\eqref{eq:cond_1_semiglobal_32}
and \eqref{eq:defbeta}, we conclude that
\begin{align*}
  |x_k| &\leq \beta\left(|x_0|,\sum_{i=0}^{k-1} T_i\right),\quad \text{for all }k\in \N_0. 
\end{align*}
Then, the system \eqref{eq:sysd} is RSS-VSR.\qed

\section*{Acknowledgments}
Work partially supported by ANPCyT grant PICT 2013-0852, Argentina.

\bibliographystyle{elsarticle-num} 
\bibliography{bibliografia201806}

\end{document}